\documentclass[prl,twocolumn,superscriptaddress,floatfix]{revtex4-2}

\usepackage{amsthm}

\newtheorem{cor}{Corollary}
\newtheorem{lemma}{Lemma}
\newtheorem{prop}{Proposition}
\newtheorem{defi}{Definition}
\newtheorem{rem}{Remark}%

\usepackage{chngcntr}

\usepackage{float}
\usepackage{bm}
\usepackage{enumerate} 
\usepackage{wrapfig}
\usepackage{braket}

\usepackage{amsmath,amsfonts,amssymb}
\usepackage{mathtools}
\usepackage{xfrac}
\usepackage{multirow}
\allowdisplaybreaks
\usepackage[colorlinks=true,linkcolor=blue,citecolor=blue,urlcolor=blue]{hyperref}
\newcommand{\smref}[1]{\ref{#1}}

\usepackage{tikz} %

\newcommand{\ci}{\rotatebox[origin=c]{90}{$\models$}}
\usepackage{centernot}
\newcommand{\notci}{\centernot{\ci}}

\usepackage[most]{tcolorbox} %
\newcommand*\mybox[1]{%
\fcolorbox{#1}{#1!50}{\rule{.07in}{0pt}\rule[-0.1ex]{0pt}{1.ex}}
}

\definecolor{redd}{RGB}{208, 2, 27}
\definecolor{bluu}{RGB}{74, 144, 226}
\definecolor{my-yellow}{RGB}{245,166,35}
\definecolor{yelloww}{RGB}{248,231,28}

\graphicspath{{./}}

\begin{document}
\ifdefined\SUPPLEMENTONLY\else

\title{Shaking the causal tree: On the faithfulness and minimality assumptions beyond pairwise interactions}

\author{Tiago Martinelli}
\email{tiago.martinelli93@gmail.com}
\affiliation{Instituto de F\'isica de S\~ao Carlos, Universidade de S\~ao Paulo, CP 369, 13560-970, S\~ao Carlos, SP, Brazil}

\author{Diogo O. Soares-Pinto}
\affiliation{Instituto de F\'isica de S\~ao Carlos, Universidade de S\~ao Paulo, CP 369, 13560-970, S\~ao Carlos, SP, Brazil}

\author{Francisco A. Rodrigues}
\affiliation{Institute of Mathematics and Computer Science, University of S\~ao Paulo, 13566-590, S\~ao Carlos, SP, Brazil}


\date[]{}

\begin{abstract}
Higher-order mechanisms can make genuine causes invisible to causal discovery, but whether this reflects fine tuning or a structural limitation remains unclear. We prove, for a binary causal family, that the failure of faithfulness is not a measure-zero pathology: it persists for every nonzero coupling vector, even though every active edge remains causally minimal. Experiments with established causal-discovery methods show how this separation propagates into algorithm-dependent failures. A Partial Information Decomposition then identifies the lattice supports probed by faithfulness and minimality, placing the two premises at opposite ends of a single conditioning hierarchy. This exposes a two-sided obstruction: synergy can erase faithfulness, while redundancy can make the positive margin protecting minimality arbitrarily weak. Exact causal premises and practical discoverability can therefore point in different directions.
\end{abstract}

\maketitle

Identifying causal parents from observational data is a central problem in science and machine learning. Constraint-based methods use faithfulness to translate conditional independencies into absent links, whereas score-based searches compare candidate graphs through a fit--complexity tradeoff closely related to causal minimality. In causal terms, faithfulness excludes independencies not encoded by the graph, while minimality imposes the complementary Occam requirement that no edge can be removed without changing the distribution \cite{spirtes,judea}. Although faithfulness is stronger, the two premises need not respond in the same way to a multivariate mechanism.

Higher-order interactions are increasingly recognised as a basic ingredient of collective dynamics, from neural and biological systems to networked populations \cite{boccaletti2023higher,malizia2024reconstructing}. Yet their causal signal need not survive projection onto pairwise statistics. In a synergistic mechanism, a genuine parent may become informative only jointly with other variables, precisely the setting in which faithfulness becomes fragile \cite{2018_Runge}. Recent neural score-based evaluations likewise identify faithfulness as a performance bottleneck \cite{olko2025}. Partial Information Decomposition (PID) \cite{2010_willbeer} provides the natural language for this problem: it separates unique, redundant, and synergistic contributions and allows us to ask which parts of a causal signal become accessible as the conditioning set grows.

Yet it remains unclear whether higher-order interactions can separate faithfulness and minimality across an entire causal family rather than at isolated parameter cancellations. In regular parametric DAG families, exact unfaithfulness is commonly confined to measure-zero algebraic hypersurfaces; geometric work nevertheless shows that violations of strong faithfulness can occupy substantial volume \cite{uhler2013geometry}. We push this further: in an odd pure-order causal family, faithfulness fails for \emph{every} nonzero coupling vector while every active edge remains causally minimal. The separation is structural, not fine tuning.

We prove this separation and then test its algorithm-specific consequences using well-known causal-discovery procedures. We also use PID to uncover the informational structure behind the two premises, identifying the lattice supports probed by faithfulness and minimality. Synergy and redundancy emerge as opposite contractions of the same conditioning hierarchy. This perspective motivates a new program for recovering causal information.

\begin{wrapfigure}[12]{r}{0.20\textwidth}
\vspace{-.45cm}
\centering
\includegraphics[width=.82\linewidth]{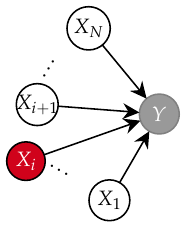}
\caption{Independent-source causal star centred at $Y$.}
\label{fig:faith}
\vspace{-.35cm}
\end{wrapfigure}

\noindent\textit{Setting and definitions.} We first consider a causal star DAG centred at a binary target $Y$, with mutually independent exogenous binary sources $\mathbf X=\{X_i\}_{i=1}^{n}$ under causal sufficiency. The target is generated by a specified response kernel with exogenous noise independent of the sources. The causal directions are part of the model specification; each active source is a parent of $Y$. We write $\mathbf X_{-i}:=\mathbf X\setminus\{X_i\}$. Random variables are denoted by capital letters and outcomes by lower-case letters. We quantify dependence using mutual information (MI) and conditional mutual information (CMI) \cite{CoverThomas}; for discrete variables, $I(X_1;Y)=0$ if and only if $X_1\ci_PY$, and $I(X_1;Y\mid X_2)=0$ if and only if $X_1\ci_PY\mid X_2$.

\noindent\textit{Causal premises.} As faithfulness and causal minimality will guide our analysis, we first express them through the conditional independences of the star in Fig.~\ref{fig:faith}:
\begin{defi}[Faithfulness \& Minimality]\label{def:global-premises}
Let $P$ be Markov to the causal star in Fig.~\ref{fig:faith}. Then:
\begin{enumerate}[(i)]
\item \emph{Faithfulness:} $P$ is globally faithful if and only if $X_i\notci_P Y\mid\mathbf Z$ for all active parents and all $\mathbf Z\subseteq\mathbf X_{-i}$, and $X_i\notci_P X_j\mid(Y,\mathbf Z)$ for all distinct active parents and all \mbox{$\mathbf Z\subseteq\mathbf X\setminus\{X_i,X_j\}$}.
\item \begin{samepage}\emph{Minimality:} $P$ is causally minimal with respect to the star if and only if $X_i\notci_P Y\mid\mathbf X_{-i}$ for every active parent.\par\end{samepage}
\end{enumerate}
\end{defi}
\begin{rem}\label{rem:classical-star}
The classical premises apply to arbitrary DAGs. Definition~\ref{def:global-premises} expresses them for the star in Fig.~\ref{fig:faith}, including dependence between active parents after conditioning on the collider $Y$. The equivalence is shown in the Supplemental Material \cite{supplemental}. Independence between an active parent and $Y$, given even one subset of the other sources, already establishes global unfaithfulness \cite{2017_Book_Scholkopf}.
\end{rem}

It is well known that a causal influence can be invisible to marginal MI. In the canonical $\mathtt{XOR}$ example, two independent unbiased sources $\{X_1,X_2\}$ determine $Y$, yet $I(X_i;Y)=0$ while $I(X_i;Y\mid\mathbf X_{-i})>0$ \cite{2010_willbeer}. Faithfulness fails while minimality survives: each arrow remains necessary, although its information is revealed only jointly.

We now move beyond this isolated example and ask whether the separation can persist throughout a causal family. The following result identifies a structural origin for marginal invisibility.

\begin{lemma}[Structural failure of faithfulness]\label{lem:lema1}
If the joint distribution is invariant under the global source flip, $P(\mathbf{x},y)=P(-\mathbf{x},y)$, then
\[
I(X_i;Y)=0
\]
for every source. Any active edge $X_i\to Y$ therefore witnesses a violation of global faithfulness: its endpoints are independent although their direct path is unblocked. (See ``Proof of Lemma~\ref{lem:lema1}'' in the Supplemental Material.)
\end{lemma}

To realise this distribution-level result in a causal family, let
\vspace{-.15cm}
\begin{equation}\label{simple_high_ising}
g_J(\mathbf{x})=\sum_{\substack{\bm{\alpha}\subseteq[n]\\|\bm{\alpha}|=k-1}}
J_{\bm{\alpha}}\prod_{i\in\bm{\alpha}}x_i
\end{equation}
be the target field, where $x_i\in\{-1,1\}$, $J_{\bm\alpha}$ are real couplings, and $k$ counts the target and its $k-1$ interacting sources. Independent unbiased binary sources and the Gibbs response $\mathbb E_J[Y\mid\mathbf X=\mathbf x]=\tanh[\beta g_J(\mathbf x)]$ \cite{1957_Jaynes} define $\mathcal F_k:=\{P_J:J\in\mathbb R^{\binom{n}{k-1}}\}$ at fixed $\beta>0$ (explicit probabilities in SM). This extends the pairwise model ($k=2$) to a fixed interaction order $k>2$. A source is active when some nonzero $J_{\bm\alpha}$ contains it; causal claims concern the active substar, while numerical profiles can retain inactive candidates (SM). For odd $k$, each product contains an even number of sources, so the symmetry required by Lemma~\ref{lem:lema1} holds throughout this coupling class:

\begin{cor}[Odd pure-order Gibbs family]\label{cor:pure-order}
For Eq.~(\ref{simple_high_ising}) with odd $k\geq3$, $g_J(-\mathbf{x})=g_J(\mathbf{x})$. Thus $I(X_i;Y)=0$ for every $J\in\mathbb{R}^{\binom{n}{k-1}}$ and every active parent $X_i$.
\end{cor}

This leaves the complementary Occam question: can a marginally invisible parent be removed? For the family just introduced, the answer follows by conditioning on all other sources and identifying exactly which couplings would have to vanish.
\begin{lemma}[Exact active-parent minimality]\label{lem:minimality-ae}
Let $2\leq k\leq n+1$ and fix $\beta>0$ and a source $X_i$ in $\mathcal{F}_k$. If $I_J(X_i;Y\mid\mathbf{X}_{-i})=0$, then every incident coefficient $J_{\bm\alpha}$ with $i\in\bm\alpha$ vanishes. Consequently, $\mathcal{Z}_i=\{J:I_J(X_i;Y\mid\mathbf{X}_{-i})=0\}$ has Lebesgue measure zero in $\mathbb{R}^{\binom{n}{k-1}}$.
\end{lemma}
Thus the exceptional set is explicit (SM): it consists precisely of coupling vectors with no interaction incident on $X_i$. For odd $k$, faithfulness fails throughout the active family. Minimality nevertheless retains its Occam role because every active arrow changes the target law.

\noindent\textit{Numerical evidence.}\label{sec:pure_baselines} The two lemmas leave a practical question: can a search reach the conditioning sets that reveal a genuine parent? We evaluate PC~\cite{spirtes}, PCMCI~\cite{2018_Runge}, and single-edge GES~\cite{judea,chickering2002optimal} with tests and scores suited to binary data, using $n=4$, 200 trials, and $N\in\{500,1000,2000,5000\}$. PC and GES use independent observations; PCMCI uses a lagged realisation of the same mechanism (SM). The pairwise control $k=2$ improves with $N$, whereas PC remains near its nominal false-positive floor for odd $k$ and GES stalls. At $k=3$, a parent surviving PCMCI's initial marginal test by chance can allow later conditioning to reveal others; this escape is strongly reduced at $\alpha=0.01$ and absent at $k=5$. Joint parent searches can access information missed by these initial marginal tests or single-edge moves. The even case $k=4$ has small nonzero marginal information and lies outside Lemma~\ref{lem:lema1} (SM).

To distinguish this obstacle from a weak total signal, we match the exact $I(\mathbf X;Y)$ of the $k=2$ and $k=3$ models in each trial. At $N=5000$ and $\alpha=0.05$, PC recall is $0.861$ for the pairwise control versus $0.058$ for the odd-order model. The minimum full-parent CMI is also larger on average in the odd-order model (SM, Table~\smref{tab:matched_signal_pc}). The gap persists even when the total signal is held fixed.

To locate the first informative conditioning set, Fig.~\ref{pure_cmi} isolates one order-$k$ interaction within Eq.~(\ref{simple_high_ising}). Any subsystem omitting one of its $k-1$ interacting sources is uninformative about $Y$; once the full block is present, CMI reveals its parents. This displays the accessibility threshold independently of a search algorithm.

\begin{figure}[!t]
\centering
\includegraphics[width=.98\columnwidth]{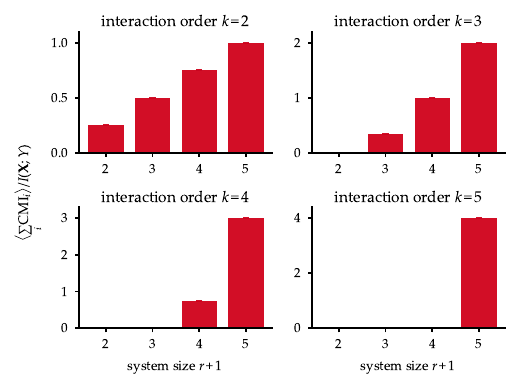}
\vspace{-.25cm}
\caption{Accessibility thresholds for one isolated order-$k$ target interaction among $n=4$ candidate sources. The abscissa $r+1$ counts an $r$-source subsystem and the target. At each rung we plot the normalized uniform average of $\sum_{X_i\in\mathbf S_r}I(X_i;Y\mid\mathbf S_r\setminus\{X_i\})$ over the $r$-source subsystems and 40 coupling draws; subthreshold rungs vanish exactly.}
\label{pure_cmi}
\vspace{-.30cm}
\end{figure}

\vspace{.05cm}
\noindent\textit{Information decomposition.}\label{sec:pid} PID addresses the issue above by decomposing the information that multiple source variables provide about a single target. Williams and Beer \cite{2010_willbeer} organised these combinations into the antichain lattice $\mathcal{A}(\mathbf X)$ (formal definitions in SM). The support identities below also apply to dependent source laws.

These combinations are called \textit{partial-information atoms} (PIs) and are defined by the mapping $\alpha\mapsto I_{\partial}(\alpha;Y)$. Decomposing MI and CMI into PIs (SM) and applying the decomposition to XOR gives, for $i=1,2$,\vspace{-.05cm}
\begin{eqnarray}
I(X_{i}; Y) &=& \text{Uni}(X_{i}; Y) + \text{Red}(X_{1}, X_{2}; Y), \\
I(X_{1};Y\mid X_{2}) &=& \text{Uni}(X_{1};Y) + \text{Syn}(X_{1},X_{2};Y)
\label{pid3},\vspace{-.4cm}
\end{eqnarray}
where $\text{Red}(X_1,X_2;Y)$ is shared by both sources, $\text{Uni}(X_1;Y)$ is uniquely available from $X_1$, and $\text{Syn}(X_1,X_2;Y)$ is revealed only by their joint state. In XOR, the dependence of $Y$ on $\{X_1,X_2\}$ is carried by the synergistic term in Eq.~(\ref{pid3}).

\noindent\textit{From PID to causal premises.} As noted in Eq.~(\ref{pid3}), conditioning can reveal synergistic information absent from marginal MI. This dissociation motivates us to investigate the conditioning operation through PID and its connection with the causal premises. We start by asking which subset $\mathcal{S}(\mathbf{X})\subset\mathcal{A}(\mathbf{X})$ isolates the synergistic atoms of the lattice and how these contributions become accessible as the conditioning set grows.

Within $\mathcal A(\mathbf X)$, singleton atoms supported on one source are unique, singleton atoms supported on $k>1$ sources are synergistic of order $k$, and nonsingleton antichains form the redundant set (formal definition in SM). Proposition~\ref{pid-cond2} connects this partition to conditioning.

\begin{prop}[Conditioning enlarges accessible synergistic support]\label{pid-cond2}
For $X_j\in\mathbf{X}$ and $\mathbf{Z}_1\subseteq\mathbf{Z}_2\subseteq\mathbf{X}_{-j}$, the set of singleton PID atoms involving $X_j$ that can contribute to $I(X_j;Y\mid\mathbf{Z}_1)$ is contained in the corresponding set for $\mathbf{Z}_2$. In particular, an atom supported on $r$ sources including $X_j$ requires $|\mathbf{Z}|\geq r-1$. This nesting concerns lattice support; it does not imply that numerical CMI increases for arbitrary distributions.
\end{prop}

For any subsystem $\mathbf S_r\subseteq\mathbf X$ with $|\mathbf S_r|=r$, define the leave-one-out contribution
\begin{equation}\label{eq:cmi-profile}
M_r(\mathbf S_r;Y):=
\sum_{X_i\in\mathbf S_r}
I\!\left(X_i;Y\mid\mathbf S_r\setminus\{X_i\}\right).
\end{equation}
For each source in the sum, we condition on the remaining $r-1$ sources of the same subsystem. Averaging $M_r$ uniformly over all $r$-source subsystems gives a profile that does not depend on a preferred ordering of the sources.

For the isolated control, the signal first appears when $r=k-1$ [Fig.~\ref{pure_cmi}]. All $k>2$ panels have zero marginal MI in this one-hyperedge case; only the odd-$k$ cancellation extends structurally to arbitrary coupling mixtures through Lemma~\ref{lem:lema1}.

At the lattice level, MI and full-parent CMI sum different PID regions (see ``Informational lattices'' in the Supplemental Material): marginal MI accesses atoms below the singleton source, whereas full-parent CMI retains atoms involving that source that are not already available from the other parents. This is a statement about lattice support and is therefore independent of the chosen redundancy functional. Claims about individually nonzero atoms additionally require a nonnegative PID, an assumption stated explicitly there. Thus synergy can support minimality while remaining absent from the marginal statistic required at the bottom rung.

\noindent\textit{Correlated-source extension.}\label{sec:mixed} To probe the top rung when predictors carry redundant signal, we retain the directed target mechanism but replace the product source prior by the pairwise source--source model
\begin{equation}\label{correlated_extension}
q_{\mathbf K}(\mathbf{x})\propto\exp\!\left[\sum_{i<j}K_{ij}x_ix_j\right],
\end{equation}
Here the parents of $Y$ may themselves share unobserved causes: Fig.~\ref{fig:faith} is no longer a complete DAG. We leave the upstream structure unspecified and study the causal-minimality margin of arrows into $Y$ at fixed $q_{\mathbf K}$ (SM). The numerical control below uses only the linear target field $g_\lambda(\mathbf x)=\lambda\sum_i x_i$: the downward profile therefore requires no higher-order target term. Following Rosas et al., we distinguish the order of the data-generating mechanism from the order of its statistical behaviour: pairwise mechanisms can generate high-order redundancy- or synergy-dominant behaviour \cite{nature}.

A solvable two-source control makes the contraction explicit. Let $X_1$ and $X_2$ be unbiased, with agreement probability $p_K=(1+e^{-2K})^{-1}$. When they agree, $Y$ copies their common sign with probability $(1+r_\lambda)/2$, where $r_\lambda=\tanh(2\lambda)$; when they disagree, $Y$ is unbiased. Thus $K$ tunes source redundancy while $\lambda$ fixes response strength, and both source--target links remain active at finite $K$ and $\lambda>0$. This is the other side of the causal coin: every target input remains necessary, but its positive minimality margin can become arbitrarily weak.

\begin{samepage}
\begin{prop}[The redundant-copy limit contracts the minimality margin]\label{prop:redundancy-margin}
Let $\psi(c)=1-h_2[(1+c)/2]$. In the two-source control,
\begin{align*}
I(X_i;Y)&=\psi(p_Kr_\lambda),\\
I(X_i;Y\mid X_j)&=p_K\psi(r_\lambda)-\psi(p_Kr_\lambda),\quad i\neq j.
\end{align*}
The marginal MI increases strictly with $K$. The CMI is positive for every finite $K$, tends to zero as $K\to\infty$, and decreases for all sufficiently large $K$.
\end{prop}
\end{samepage}
Thus redundancy improves the marginal endpoint while making the minimality test arbitrarily weak without violating it. The proof and asymptotics are given in SM.

\begin{figure}[t]
\centering
\includegraphics[width=.98\columnwidth]{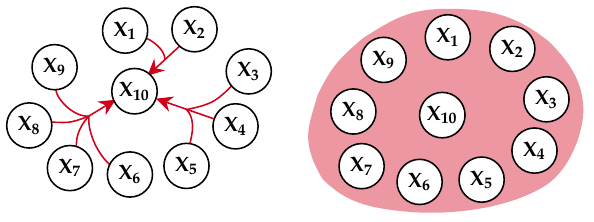}\\[-.10cm]
\includegraphics[width=.98\columnwidth]{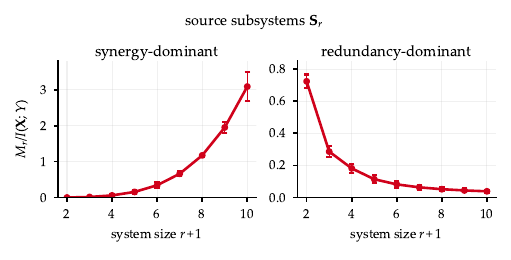}
\vspace{-.25cm}
\caption{Opposite conditioning trajectories. Top: the original schematic uses $X_{10}\equiv Y$; the cloud denotes source dependence in $q_{\mathbf K}$, with its upstream causal structure unspecified. Bottom: exact normalized leave-one-out profiles, averaged uniformly over source subsystems and 40 coupling draws, for independent sources with exclusive $k=3,4,5$ target hyperedges (left) and the finite pairwise correlated-copy control with a linear target field (right). The redundant endpoint is positive but visually approaches zero.}
\label{growth_cmi}
\vspace{-.30cm}
\end{figure}

Figures~\ref{pure_cmi} and \ref{growth_cmi} average Eq.~(\ref{eq:cmi-profile}) uniformly over the subsystems of each size and divide by the total target information $I(\mathbf{X};Y)$ within each coupling draw. In the isolated and sparse-exclusive mechanisms, the normalized profile grows with $r$. For odd $k$, Lemma~\ref{lem:lema1} fixes marginal parent information at zero while Lemma~\ref{lem:minimality-ae} guarantees a nonzero full-parent rung for every active parent. This growth belongs to the controlled mechanisms above; it is not a universal monotonicity theorem.

In the finite correlated-copy control, marginal parent--target information returns while, over the coupling window used in Fig.~\ref{growth_cmi}, the normalized profile decreases monotonically from $0.72$ to the $10^{-2}$ scale (SM). The endpoint remains positive, so this is not an exact minimality violation; the contracting margin can instead fall below finite-sample resolution. Redundancy can therefore improve bottom-rung detectability while making the top-rung Occam test fragile. Neither endpoint is uniformly informative across the two regimes.

Source coupling changes this picture: it can inflate marginal parent--target information while suppressing the information added by conditioning, yielding a \emph{redundancy-dominant} profile. This does not identify a particular PID atom. Source redundancy is distinct from mechanistic redundancy in stricter taxonomies and can depend on the chosen functional \cite{harder2013,2018_Crutch,2017_Ince}. The full profile extends the familiar three-variable comparison $I(X_i;Y\mid Z)-I(X_i;Y)$ across conditioning scales \cite{2008_network_coding}.

We aggregate the marginal and full-parent tests as
\vspace{-.2cm}
\begin{equation}\label{eq:FM}
F \;=\; \sum_{i} I(X_i;Y),
\qquad
M \;=\; \sum_{i} I\!\left(X_i;Y \mid \mathbf{X}_{-i}\right),
\end{equation}
and define a new directed metric as their signed separation:
\begin{equation}\label{eq:omegaD}
\Omega_{\mathrm D}(\mathbf X;Y):=F-M.
\end{equation}
Here, $F$ and $M$ sum the marginal and full-parent conditional information. For the profile in Eq.~(\ref{eq:cmi-profile}) before normalisation, the average first rung is $F/n$, while the full rung is $M_n(\mathbf X;Y)=M$. Their PID atom supports are illustrated in Supplemental Fig.~\smref{fig_MI&CMI}. They do not replace parent-wise tests: global faithfulness requires $\min_i I(X_i;Y)>0$, while the arrows into $Y$ satisfy the minimality criterion if $m_{\min}:=\min_i I(X_i;Y\mid\mathbf{X}_{-i})>0$; both minima concern active parents. For a nonempty active parent set, $F=0$ and $M=0$ characterise maximal failure at the respective endpoints.

The directed, target-aware metric $\Omega_{\mathrm D}$ records which aggregate causal-premise margin dominates. By summing over sources, it links the collective organisation of information to the statistical margins of the two causal premises. Its systematic PID and access-geometric properties are under development and will be presented in future work.

Combining Lemmas~\ref{lem:lema1} and~\ref{lem:minimality-ae} with Proposition~\ref{prop:redundancy-margin}, we obtain the complementary limits in Supplemental Corollary~\smref{cor:opposite-contractions}. These motivate the following definition.

\begin{defi}\label{def:directed-dominance}
For the directed balance \mbox{$\Omega_{\mathrm D}=F-M$}, we call the source--target organisation \emph{redundancy-dominated} when $\Omega_{\mathrm D}>0$ and \emph{synergy-dominated} when $\Omega_{\mathrm D}<0$.
\end{defi}

In this interpretation, synergy can suppress the marginal evidence required by faithfulness, whereas redundancy can weaken the conditional evidence protecting causal minimality.

Our measure resembles O-information \cite{Rosas2019OInfo} in that its sign expresses a balance. Here, however, the balance is linked to the marginal and conditional evidence entering the causal premises; its sign alone does not certify that either premise fails. Each top-rung CMI is also the KL separation from the corresponding reduced target kernel (Supplemental Eq.~\smref{eq:edge-deletion-kl}), and both endpoints must still be tested per parent.

\noindent\textit{Conclusion.} We have established a structural separation between two foundations of causal discovery: in pure-order Gibbs models with odd $k\geq3$, $I(X_i;Y)$ vanishes \emph{identically} although every active parent remains causally minimal. This is not a measure-zero artefact: the cancellation holds throughout coupling space, and no perturbation of its couplings can restore marginal visibility. Experiments with PC, GES, and PCMCI, using tests and scores suited to binary observations, show how this obstruction reaches different inference strategies. The consequences depend on the search procedure; exhaustive group or hypergraph searches are not excluded.

Our PID-based analysis further identifies two qualitatively distinct losses of causal resolution. Synergistic organisation can erase marginal visibility while preserving minimality; redundancy can make the positive conditional margin arbitrarily weak on an open finite-parameter region. The exact results anchor $\Omega_{\mathrm D}=F-M$ and motivate this interpretation. Faithfulness and minimality probe opposite ends of one information structure: discoverability can deteriorate at either end while the causal statement remains exact.

These results point to a natural question: at what scale of aggregation does a collective causal signal become visible? The joint state of a source block $B$, $C_B=(X_i)_{i\in B}\in\prod_{i\in B}\mathcal X_i$, retains its constituents' information without presupposing their numerical product. In the isolated one-interaction control, the known-block map $Z_B=\prod_{i\in B}X_i$ gives an explicit coarse-graining: it is sufficient for $Y$ and has $I(Z_B;Y)>0$ for nonzero coupling (see ``Illustration for the conclusion'' in the Supplemental Material). The block is supplied, not learned. The structural result identifies the cause; coarse-graining asks at what scale it can be resolved. This programme connects with multiscale causal research \cite{erik2013,Rosas_2020_2} and complements searches at the original scale \cite{2012_Runge,faes2022}; learning and validating general maps remain future work.

\begin{acknowledgments}
We thank Dragan Ma\v{s}ulovi\'{c} for the fruitful discussion, P. Mediano for providing the codes, and T. Per\'{o}n and K. Roster for reading the manuscript. This work was funded by CNPq (Grant No.~140665/2018-8) and FAPESP (Grant No.~2018/12072-0).
\end{acknowledgments}

\fi

\ifdefined\LETTERONLY\else

\clearpage
\onecolumngrid
\raggedbottom
\begin{center}
{\large\bfseries Supplemental Material}\\[.3em]
{Shaking the causal tree: On the faithfulness and minimality assumptions beyond pairwise interactions}
\end{center}

\appendix
\setcounter{secnumdepth}{1}

\section{Proofs \label{appA}}

We first establish the structural separation between faithfulness and minimality, then derive the redundant-copy limit and combine these results in the causal-premise balance.

\subsection*{Proof of Lemma~\ref{lem:lema1}}

\begin{proof}
Marginalising the assumed joint invariance over $\mathbf{x}_{-i}$ gives
\[
P(X_i{=}x_i,Y{=}y) = \!\!\!\sum_{\mathbf{x}_{-i}}\!\!\! P(\mathbf{x},y) = \!\!\!\sum_{\mathbf{x}_{-i}}\!\!\! P(\tau\mathbf{x},y) = P(X_i{=}{-x_i},Y{=}y),
\]
where $\mathbf{x}_{-i}\mapsto-\mathbf{x}_{-i}$ is a bijection. Thus $X_i$ is uniform even conditional on $Y$, so $X_i\ci_PY$. For an active edge $X_i\to Y$, the direct path is unblocked even without conditioning. This statistical independence therefore also contradicts global faithfulness.
\end{proof}

\subsection*{Proof of Corollary~\ref{cor:pure-order}}

\begin{proof}
For odd $k$, every product in $g_J$ contains the even number $k-1$ of sources. Hence $g_J(-\mathbf{x})=g_J(\mathbf{x})$. The Gibbs closure gives the same conditional mean at $\mathbf x$ and $-\mathbf x$ and, since $Y$ is binary, the same conditional law. The independent Rademacher prior is also flip-invariant, so Lemma~\ref{lem:lema1} applies.
\end{proof}

\subsection*{Proof of Lemma~\ref{lem:minimality-ae}}

\begin{proof}
Fix a configuration $\mathbf{x}_{-i}$ and write
$q_{\pm}=P_J(Y=1\mid X_i=\pm1,\mathbf{x}_{-i})$. Since the sources are independent Rademacher variables, the contribution of this context to
$I_J(X_i;Y\mid\mathbf{X}_{-i})$ is
\[
h_2\!\left(\frac{q_++q_-}{2}\right)
-\frac{h_2(q_+)+h_2(q_-)}{2}\geq0.
\]
Strict concavity of binary entropy makes equality possible only if $q_+=q_-$. Hence zero full-parent CMI forces this equality for every $\mathbf{x}_{-i}$. For $\beta>0$, the map $g\mapsto\tanh(\beta g)$ is injective, so
\[
g_J(+1,\mathbf{x}_{-i})-g_J(-1,\mathbf{x}_{-i})
=2\sum_{\bm\alpha\ni i}J_{\bm\alpha}
\prod_{j\in\bm\alpha\setminus\{i\}}x_j=0
\]
on the whole Boolean cube. The parity monomials are orthogonal under the uniform measure. Multiplying the last identity by the monomial indexed by
$\bm\alpha\setminus\{i\}$ and averaging over $\mathbf{x}_{-i}$ therefore isolates $J_{\bm\alpha}$ and gives $J_{\bm\alpha}=0$ for every $\bm\alpha\ni i$. In particular, for any fixed incident $\bm\alpha_0$,
$\mathcal{Z}_i\subseteq\{J:J_{\bm\alpha_0}=0\}$, a coordinate hyperplane of Lebesgue measure zero.
\end{proof}

\subsection*{Proof of Proposition~\ref{prop:redundancy-margin}}

\begin{proof}
For completeness, the equivalent Gibbs parametrisation is
$q_K(x_1,x_2)=e^{Kx_1x_2}/(4\cosh K)$ and
$P_\lambda(y\mid x_1,x_2)=e^{\lambda y(x_1+x_2)}/\{2\cosh[\lambda(x_1+x_2)]\}$.
Write $p=p_K=P(X_1=X_2)$ and $r=r_\lambda$. When the sources agree, $P(Y=X_i\mid X_1=X_2)=(1+r)/2$; when they disagree, the target is uniform. Hence $X_i$ and $Y$ are unbiased with correlation $pr$, which gives
\[
I(X_i;Y)=\psi(pr).
\]
Moreover,
\[
H(Y\mid X_1,X_2)=p\,h_2[(1+r)/2]+(1-p),
\]
whereas $H(Y\mid X_j)=h_2[(1+pr)/2]$. Their difference yields the stated CMI.

On $0<c<1$, $\psi(0)=0$, $\psi'(c)=\tfrac12\log_2[(1+c)/(1-c)]>0$, and
$\psi''(c)=[\ln 2\,(1-c^2)]^{-1}>0$. Since $p_K$ increases with $K$, the marginal MI increases strictly. Strict convexity gives $\psi(pr)<p\psi(r)$ for $0<p<1$, proving positive CMI at every finite $K$. As $K\to\infty$, $p\to1$ and the CMI tends to zero. Its derivative with respect to $p$ is $\psi(r)-r\psi'(pr)$; at $p=1$ this is negative because strict convexity implies $\psi(r)<r\psi'(r)$. Continuity therefore gives eventual monotonic decrease. Finally, the first-order coefficient is
\[
\lim_{K\to\infty}
\frac{I(X_i;Y\mid X_j)}{1-p_K}
=r\psi'(r)-\psi(r)
=-\frac{1}{2}\log_2(1-r^2)>0.
\]
Together with $1-p_K=\mathcal{O}(e^{-2K})$, this gives
$I(X_i;Y\mid X_j)=\mathcal{O}(e^{-2K})$ and establishes exponential contraction of the minimality margin.
\end{proof}

\subsection*{Exact evidence for the causal-premise balance}

\begin{cor}[Exact anchors of the causal-premise balance]\label{cor:opposite-contractions}
In the odd pure-order family, $F=0$ throughout coupling space, whereas every active parent has positive full-parent CMI; in particular, $M>0$ for every nonzero coupling vector. Conversely, for every $\epsilon>0$, the redundant-copy family contains a nonempty open set $U_\epsilon\subset(0,\infty)^2$ of finite $(K,\lambda)$ and a constant $c_\epsilon>0$ such that $F>c_\epsilon$ and $0<M<\epsilon$ throughout $U_\epsilon$.
\end{cor}

\begin{proof}
The first statement follows from Lemmas~\ref{lem:lema1} and \ref{lem:minimality-ae}: any nonzero pure-order coefficient activates at least one parent and hence contributes a positive top-rung term. For the second, fix $\lambda>0$. Proposition~\ref{prop:redundancy-margin} gives $M(K,\lambda)>0$ at finite $K$ and $M(K,\lambda)\to0$, while $F(K,\lambda)$ has a positive limit. Choose a finite point with $0<M_0<\epsilon$ and $F_0>0$, then shrink an open neighbourhood inside $(0,\infty)^2$ so that $M<\epsilon$ and $F>F_0/2$. Taking $c_\epsilon=F_0/2$ gives the uniform bound, while finite positive parameters ensure $M>0$. Hence an arbitrarily small positive causal-minimality margin for the arrows into $Y$ occurs on an open finite-parameter region.
\end{proof}

\section{Causal definitions and dynamical realisations \label{appCausal}}

We now connect the statistical criteria used in the main text to the graph-based causal premises and specify their model realisations. The temporal construction is needed only for the PCMCI experiment.

\begin{defi}[Causal model and causal Markov condition]
A causal model is a pair $(G,P)$ in which $G=(V,E)$ is a DAG whose arrows represent direct causal relations and $P$ is Markov with respect to $G$. Equivalently, its observational distribution factorises as
\[
P(\mathbf v)=\prod_{V_j\in V}P(v_j\mid\mathrm{PA}_j).
\]
Here $\mathrm{PA}_j$ denotes the parents of $V_j$ in $G$. We assume causal sufficiency for the modeled variables unless a section is explicitly identified as a sensitivity analysis.
\end{defi}

\begin{defi}[Global faithfulness and causal minimality]
For disjoint node sets $A,B,S\subseteq V$, $P$ is faithful to $G$ when
\[
A\ci_P B\mid S \quad\Longleftrightarrow\quad
A\ci_G B\mid S,
\]
where $\ci_G$ denotes $d$-separation. Under the causal Markov condition, faithfulness adds the implication from statistical independence to $d$-separation. The distribution $P$ is causally minimal with respect to $G$ when it is Markov to $G$ but not to any proper subgraph obtained by deleting one or more arrows.
\end{defi}

\paragraph{Equivalence of the main-text star formulation.}
The $d$-connected singleton queries in the star are exactly an active parent and $Y$ given any other sources, and two distinct active parents given $Y$ and any other sources. The direct path gives the first class; the collider path through $Y$ gives the second, since $Y$ has no descendants. Inactive sources are isolated. Global faithfulness therefore implies both clauses in Definition~\ref{def:global-premises}-(i). Conversely, if disjoint node sets are not $d$-separated by a conditioning set, some cross-pair is $d$-connected given that same set. Its dependence is required by one of the two clauses. Conditional independence of the node sets would imply independence of this pair by decomposition, a contradiction. The Markov assumption supplies the reverse implication, establishing the equivalence with global faithfulness. This uses only decomposition of set independence, not the generally invalid converse from pairwise to set independence.

We next ask when an arrow is indispensable in a Markov DAG. The resulting parent-wise test is the standard causal-minimality criterion, not an additional premise \cite[Proposition~6.36]{2017_Book_Scholkopf}.
\begin{prop}[Equivalence of causal minimality]
\label{prop:general-minimality-equivalence}
Let $P$ admit a density with respect to a product measure and be Markov to a DAG $G$. Then $P$ is causally minimal with respect to $G$ if and only if
\[
V_j\notci_P V_i\mid\mathrm{PA}_j\setminus\{V_i\}
\]
for every node $V_j$ and every parent $V_i\in\mathrm{PA}_j$. For finite variables, this is equivalent to strictly positive CMI for each arrow.
\end{prop}
\begin{proof}
If $P$ is Markov to a proper edge-deleted subgraph, it is also Markov to $G$ with at least one missing arrow $V_i\to V_j$ deleted. The local Markov property then gives $V_j\ci_P V_i\mid\mathrm{PA}_j\setminus\{V_i\}$. Conversely, this conditional independence allows the factor $p(v_j\mid\mathrm{pa}_j)$ to be replaced by $p(v_j\mid\mathrm{pa}_j\setminus\{v_i\})$ almost surely. The remaining factors are unchanged, so $P$ factorises over the graph without $V_i\to V_j$. The edge is therefore removable. Applying the argument to every arrow establishes the equivalence.
\end{proof}

The model used in the main text has independent exogenous sources
$\mathbf X=\{X_i\}_{i=1}^{n}$ and a target $Y$ generated by its specified
conditional mechanism with independent exogenous noise. Its observational law
factorises as
\begin{equation}
P(\mathbf{x},y)=\prod_{i=1}^{n}q_i(x_i)\,P(y\mid\mathbf{x}),
\label{eq:causal-factorisation}
\end{equation}
An edge $X_i\to Y$ is active when the kernel $P(y\mid\mathbf{x})$ depends on
$x_i$ for at least one value of $\mathbf{x}_{-i}$.

For the pure-order family, an explicit structural causal model is
\[
X_i=U_i,\qquad
Y=2\mathbf 1\!\left\{U_Y\leq
\frac{1+\tanh[\beta g_J(\mathbf X)]}{2}\right\}-1,
\]
where the $U_i$ are independent Rademacher variables and
$U_Y\sim\mathrm{Uniform}(0,1)$ is independent of them. This specifies the
causal directions and gives exactly the response kernel used in the proofs
and simulations, without requiring observations at different times.

For sparse or boundary coupling vectors, write
$A(J)=\{i:\exists\bm\alpha\ni i,\ J_{\bm\alpha}\ne0\}$.
The full candidate layer remains $\mathbf X$, while the active causal graph has
$\mathrm{PA}_Y=\{X_i:i\in A(J)\}$. Definition~\ref{F&M} uses the all-active
convention; inactive independent candidates can be included in its conditioning
sets without changing the premises for an active parent. All reported minima
are over active parents. At $J=0$ the active graph has no arrows; global
unfaithfulness of that graph is claimed only for $J\ne0$.

For the PID statements that involve all conditioning sets, we also record the
following auxiliary conditions. They retain the parent-target part of
Def.~\ref{def:global-premises}, without imposing its collider-conditioned
source-pair tests.
\begin{defi}\label{F&M}
For a target $Y$ with parent set $\mathrm{PA}_Y=\mathbf X$, consider the following conditions on each active link $X_i\to Y$:

(CF) $X_i\notci_PY\mid\mathbf Z$ for every $\mathbf Z\subseteq\mathbf X_{-i}$;

(CM) $X_i\notci_PY\mid\mathbf X_{-i}$.
\end{defi}

Applying Proposition~\ref{prop:general-minimality-equivalence} to the target gives the criterion used throughout our analysis. For the independent-source family it characterises minimality of the complete star:
\begin{prop}[Minimality in the causal star]
\label{prop:minimality-equivalence}
Assume that Eq.~(\ref{eq:causal-factorisation}) has full support. The edge
$X_i\to Y$ is removable without changing the observational
distribution if and only if
\[
X_i\ci_PY\mid\mathbf{X}_{-i}.
\]
Consequently, the star DAG is causally minimal if and only if
$I(X_i;Y\mid\mathbf{X}_{-i})>0$ for every parent $X_i$.
\end{prop}

\begin{proof}
If $X_i\ci_PY\mid\mathbf{X}_{-i}$, then
$P(y\mid\mathbf{x})=P(y\mid\mathbf{x}_{-i})$ at every state of positive
probability. The factorisation is therefore unchanged after deleting
$X_i\to Y$. Conversely, if the edge can be deleted without changing $P$, the
target kernel admits the reduced factorisation
$P(y\mid\mathbf{x})=P(y\mid\mathbf{x}_{-i})$, which is precisely
$X_i\ci_PY\mid\mathbf{X}_{-i}$. For discrete variables, conditional
independence is equivalent to zero CMI.
\end{proof}

Independence among the parents is not required by Proposition~\ref{prop:general-minimality-equivalence}. For the correlated-source extension, its causal interpretation concerns the arrows into $Y$ in a compatible Markov causal model whose upstream structure generates $q_{\mathbf K}$, whose target parents are $\mathbf X$, and whose target noise remains independent of that upstream system. The complete graph may include further observed or latent causes; neither their identification nor minimality of all upstream arrows is asserted. No change to the target mechanism or to the source law is needed for this interpretation.

The minimality margin also has an operational interpretation. For the reduced-target-kernel projection
\[
Q_i(\mathbf{x},y):=P(\mathbf{x})P(y\mid\mathbf{x}_{-i}),
\]
a direct expansion gives
\begin{equation}\label{eq:edge-deletion-kl}
D_{\mathrm{KL}}(P\Vert Q_i)=I(X_i;Y\mid\mathbf{X}_{-i}).
\end{equation}
Thus the top-rung CMI is the information per observation separating the full distribution from this projection. In the independent-source star it is the edge-deleted model. The identity also holds for an arbitrary source law $q(\mathbf x)$, including the correlated-source extension, and measures the causal-minimality margin of the arrow into $Y$ under the compatible causal interpretation above. The source law is held fixed; no assertion about the minimality of its upstream structure follows. We use Eq.~\eqref{eq:edge-deletion-kl} only to interpret the margin; finite-sample resolution bounds and representation selection are beyond the present scope.

\begin{prop}[Scope of condition (CF)]
\label{prop:local-global}
Restrict discovery queries to $X_i$ versus $Y$ conditional on
$\mathbf{Z}\subseteq\mathbf{X}_{-i}$. Then Def.~\ref{F&M}-(CF) is
exactly faithfulness restricted to this admissible query class. It is not
equivalent to unrestricted global faithfulness.
\end{prop}

\begin{proof}
For a direct edge $X_i\to Y$, no set of other sources blocks
the edge. Restricted faithfulness therefore requires
$X_i\notci_PY\mid\mathbf{Z}$ for every admissible $\mathbf{Z}$, exactly as in
Def.~\ref{F&M}-(CF). Thus global faithfulness implies (CF), and any violation
of (CF) is a violation of global faithfulness. Unrestricted global faithfulness also tests other node
pairs and conditioning sets, including source pairs conditioned on the collider
$Y$; these statements lie outside the admissible query class.
\end{proof}

Choosing $\mathbf{Z}=\mathbf{X}_{-i}$ in (CF) gives the minimality condition of
Prop.~\ref{prop:minimality-equivalence}. Thus restricted faithfulness implies
causal minimality of the independent-source star, whereas the reverse implication fails for XOR and for the
odd pure-order family.

\paragraph{Independent sources and the scope of (CF).}
For a product source law and $\mathbf Z\subseteq\mathbf X_{-i}$, the chain rule gives
\[
I(X_i;Y\mid\mathbf Z)=I(X_i;Y)+I(X_i;\mathbf Z\mid Y),
\]
because $I(X_i;\mathbf Z)=0$. Hence (CF) is equivalent to positive marginal MI in this independent-source setting. More generally, if $\mathbf Z_1\subseteq\mathbf Z_2$, their CMI difference is
\[
I(X_i;Y\mid\mathbf Z_2)-I(X_i;Y\mid\mathbf Z_1)
=I(X_i;\mathbf Z_2\setminus\mathbf Z_1\mid Y,\mathbf Z_1)\geq0,
\]
since $X_i$ is independent of the added sources conditional on $\mathbf Z_1$. This monotonicity follows from the product prior and is not a synergy criterion. These simplifications are not valid for arbitrary dependent source laws, which are allowed in the PID support statements below.

\subsection*{Dynamical realisation for PCMCI}
The causal mechanism above also admits a dynamical realisation. This is the
format used for PCMCI; PC and GES instead use independent observations of the
same source--target law.
\begin{defi}[Dynamical causal process]
A dynamical causal process is a causal model whose nodes are partitioned into time layers $V_t$ and whose arrows point strictly forward in time. Thus there are no instantaneous arrows, and each transition factorises into kernels conditioned on parents in earlier layers. For a first-order process,
\[
P(\mathbf v_{0:T})=P(\mathbf v_0)
\prod_{t=1}^{T}\prod_{V_{j,t}\in V_t}
P(v_{j,t}\mid\mathrm{PA}_{j,t}),
\qquad \mathrm{PA}_{j,t}\subseteq V_{t-1}.
\]
\end{defi}
For each transition we draw fresh independent sources $\mathbf X(t-1)$ and
generate $Y(t)$ from the same Gibbs response. The pairs
$(\mathbf X(t-1),Y(t))$ therefore have exactly the joint law studied above.
Time indices specify this experimental realisation; neither the symmetry
argument, the active-parent minimality proof nor the PID identities require
chronological observations.

\section{\label{appB}Informational lattices}

\noindent To distinguish the partial-information contributions to MI, we follow Williams and Beer \cite{2010_willbeer} and index the decomposition by collections in which no source group contains another. Formally,

\begin{defi}[The informational antichain]\label{wb_chain}
Consider the ensemble of variables $\mathbf{X}=\{X_{1},X_{2}, \ldots, X_{n}\}$, which are informational sources to the target $Y$. The set
\begin{equation}
\begin{gathered}
\mathcal{A}(\mathbf{X}) = \big\{\alpha \in \mathcal{P}^{+}(\mathcal{P}^{+}(\mathbf{X})):
A_{1}\neq A_{2}\Rightarrow A_{1}\nsubseteq A_{2},\,
\forall A_{1}, A_{2} \in \alpha\big\},
\end{gathered}
\end{equation}
where $\mathcal{P}^{+}(S) = \mathcal{P}(S)\setminus \{\emptyset\}$ denotes the set of nonempty subsets of $S$. This is the informational antichain of the sources $\mathbf{X}$. Henceforth, we denote elements of $\mathcal{A}(\mathbf{X})$ by omitting inner brackets, separating source groups by a dot, and concatenating the indices within each group. For example, $X_{1}\cdot X_{2}X_{3}$ represents the antichain $\{\{X_{1}\},\{X_{2},X_{3}\}\}$ \cite{2016_Crutch}.
\end{defi}

\begin{defi}[The unique, redundant, and synergistic sets]\label{orders}
Let $\mathbf{X}=\{X_i\}_{i=1}^{n}$ be a set of sources and $\mathbf{Z}\subseteq\mathbf{X}$. Write
\[
\mathcal{A}_{\mathbf{Z}}:=\{\alpha\in\mathcal{A}(\mathbf{X}):
A\subseteq\mathbf{Z}\ \text{for every }A\in\alpha\}.
\]
For $k\in[n]$, define:
\begin{enumerate}[(a)]
\item the \textbf{unique set}, $\mathcal{U}(\mathbf{Z}):=\{\{A\}\in\mathcal{A}_{\mathbf{Z}}:|A|=1\}$;
\item the \textbf{synergistic set of order $k>1$}, $\mathcal{S}^{(k)}(\mathbf{Z}):=\{\{A\}\in\mathcal{A}_{\mathbf{Z}}:|A|=k\}$;
\item the \textbf{redundant set of order $k$}, $\mathcal{R}^{(k)}(\mathbf{Z}):=\{\alpha\in\mathcal{A}_{\mathbf{Z}}:|\alpha|>1,\ |A|=k\ \text{for some }A\in\alpha\}$.
\end{enumerate}
When all orders are considered, we omit the superscript and write $\mathcal{S}(\mathbf{X})=\bigcup_k\mathcal{S}^{(k)}(\mathbf{X})$ and $\mathcal{R}(\mathbf{X})=\bigcup_k\mathcal{R}^{(k)}(\mathbf{X})$. Thus $\mathcal{A}(\mathbf{X})=\mathcal{U}(\mathbf{X})\sqcup\mathcal{S}(\mathbf{X})\sqcup\mathcal{R}(\mathbf{X})$; see Figs.~\ref{fig_mydefs} and \ref{fig_cond}.
\end{defi}

Antichains form a lattice \cite{birkhoff}, with the natural partial order

\begin{equation}
\forall \alpha, \beta \in \mathcal{A}(\mathbf{X}),\; \alpha \preceq \beta \iff (\forall B\in \beta:\; \exists A\in \alpha,\; A \subseteq B).
\end{equation}
When ascending the lattice, the cumulative redundancy function $I_{\cap}(\alpha;Y)$ increases monotonically \cite{2010_willbeer}. M\"obius inversion of $I_{\cap}$ defines the partial-information function $I_{\partial}$, which assigns an atom to each element of $\mathcal{A}(\mathbf{X})$. For a collection of sources $\alpha \in \mathcal{A}(\mathbf{X})$, it is defined recursively by
\begin{equation}
I_{\partial}(\alpha; Y) = I_{\cap}(\alpha; Y) - \sum_{\beta \prec \alpha}I_{\partial}(\beta; Y).
\label{mobius}
\end{equation}
For a singleton source group, self-redundancy identifies $I_{\cap}(\{A\};Y)=I(A;Y)$ and closes the recursion. If $q=|\mathbf{X}|$, the number of PI atoms is $|\mathcal{A}(\mathbf{X})|=M(q)-2$, where $M(q)$ is the $q$th Dedekind number \cite{birkhoff}; for $q=1,2,3,4,5,\ldots$, this gives $1,4,18,166,7579,\ldots$.

\subsection{\label{appB1}Conditioning antichains}

We now ask which atoms survive conditioning. Even with three sources, a redundant atom can remain in the CMI. For $\mathbf{X}=\{X_{1},X_{2},X_{3}\}$,
\begin{eqnarray}\label{eq:condPiD}
I(X_{1}; Y|X_{2},X_{3}) &=& I(X_{1},X_{2},X_{3}; Y) - I(X_{2},X_{3}; Y) \\ 
&=& I_{\partial}(X_{1};Y)+I_{\partial}(X_{1}X_{2};Y)+I_{\partial}(X_{1}X_{3};Y) \nonumber\\
&&{}+I_{\partial}(X_{1}X_{2}X_{3};Y)+I_{\partial}(X_{1}X_{2}\cdot X_{1}X_{3};Y).
\end{eqnarray}
Equation~\eqref{eq:condPiD} shows that the last, redundant atom is not eliminated by conditioning; see Fig.~\ref{fig_cond}. It lies outside the down-set\footnote{The down-set of $\alpha$ is $\mathord{\downarrow}\alpha:=\{\beta\in\mathcal{A}(\mathbf{X}):\beta\preceq\alpha\}$.} of the joint source group $\{X_{2},X_{3}\}$. The same support also contains the unique atom and synergistic singleton atoms of every available order. Proposition~\ref{pid-cond} formalizes this operation.

\subsection{Graphical illustration of Definition~\ref{orders}}%
\vspace{-.75cm}
\begin{figure}[H]
\begin{minipage}[b]{0.45\textwidth}
    \includegraphics[width=7cm, height=9.5cm]{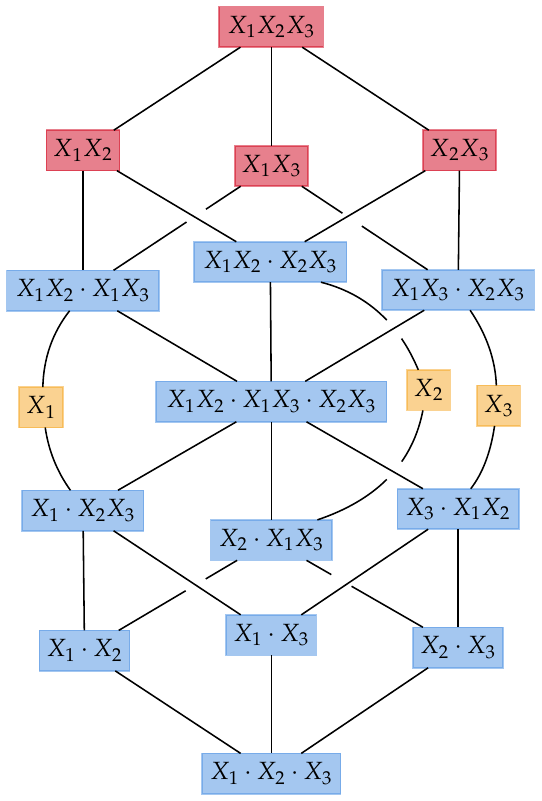}
    \caption{Illustration of Def.~\ref{orders} on the informational lattice $\mathcal{A}(\mathbf{X})$ for $\mathbf{X}=\{X_{1},X_{2},X_{3}\}$. Colors indicate atoms in \mybox{bluu} $\mathcal{R}(\mathbf{X})$, \mybox{my-yellow} $\mathcal{U}(\mathbf{X})$, and \mybox{redd} $\mathcal{S}(\mathbf{X})$.}%
    \label{fig_mydefs}
\end{minipage}
\hfill
\begin{minipage}[b]{0.45\textwidth}
    \includegraphics[width=7cm, height=9.5cm]{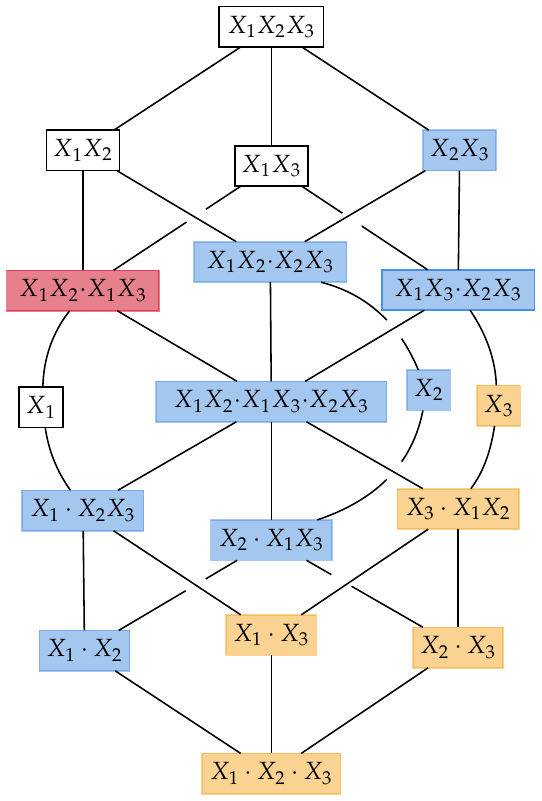}
    \caption{Effect of conditioning $I(X_{1};Y\mid\mathbf{Z})$ on the lattice $\mathcal{A}(\mathbf{X})$ for $\mathbf{X}=\{X_{1},X_{2},X_{3}\}$. The \mybox{my-yellow} region is removed for $\mathbf{Z}=\{X_{3}\}$; the \mybox{my-yellow}+\mybox{bluu} regions are removed for $\mathbf{Z}=\{X_{2},X_{3}\}$; the \mybox{redd} region remains in the latter CMI.}
    \label{fig_cond}
\end{minipage}

\end{figure}

\subsection{PID support of MI and CMI}

\begin{prop}[PID view of conditioning]\label{pid-cond}
Let $X\in\mathbf{X}$ and $\mathbf{Z}\subseteq\mathbf{X}\setminus\{X\}$. In the ambient PID lattice $\mathcal{A}(\mathbf{X})$,
\begin{align}
I(X;Y\mid\mathbf{Z})
&= I(X,\mathbf{Z};Y)-I(\mathbf{Z};Y)\nonumber\\
&=\sum_{\alpha\in
\mathord{\downarrow}\{\{X\}\cup\mathbf{Z}\}
\setminus\mathord{\downarrow}\{\mathbf{Z}\}}
I_{\partial}(\alpha;Y),
\label{eq:pid-conditioning}
\end{align}
where $\{\mathbf{Z}\}$ denotes the singleton source group containing all variables in $\mathbf{Z}$; for $\mathbf{Z}=\emptyset$, its down-set is empty.
\end{prop}
\vspace{.1cm}

\begin{prop}[MI and full-parent CMI supports]\label{mi&cmi_lattices}
For $X_i\in\mathbf{X}$, define
\begin{align}
\mathcal{L}_{\mathrm{MI}}(X_i)
&:=\mathord{\downarrow}\{\{X_i\}\},\nonumber\\
\mathcal{L}_{\mathrm{CMI}}(X_i)
&:=\mathcal{A}(\mathbf{X})\setminus
\mathord{\downarrow}\{\mathbf{X}_{-i}\}.
\label{eq:mi-cmi-supports}
\end{align}
Then
\begin{align}
I(X_i;Y)&=\sum_{\alpha\in\mathcal{L}_{\mathrm{MI}}(X_i)}I_{\partial}(\alpha;Y),
\label{mi_cmi_atoms1}\\
I(X_i;Y\mid\mathbf{X}_{-i})
&=\sum_{\alpha\in\mathcal{L}_{\mathrm{CMI}}(X_i)}I_{\partial}(\alpha;Y).
\label{mi_cmi_atoms2}
\end{align}
The first support consists precisely of antichains containing the singleton group $\{X_i\}$. The second consists precisely of antichains for which every source group contains $X_i$.
\end{prop}
\begin{proof}
Self-redundancy gives the first identity as the cumulative PID at $\{\{X_i\}\}$. By the lattice order, $\alpha\preceq\{\{X_i\}\}$ if and only if some group $A\in\alpha$ satisfies $A\subseteq\{X_i\}$, equivalently if $\{X_i\}\in\alpha$.

For the second identity, apply Proposition~\ref{pid-cond} with $\mathbf{Z}=\mathbf{X}_{-i}$. The first down-set is then all of $\mathcal{A}(\mathbf{X})$. Moreover, $\alpha\preceq\{\mathbf{X}_{-i}\}$ if and only if at least one group $A\in\alpha$ excludes $X_i$. Its complement therefore contains exactly those antichains in which every group contains $X_i$.
\end{proof}

\begin{figure}[htb]
\begin{center}
\includegraphics[width=7cm]{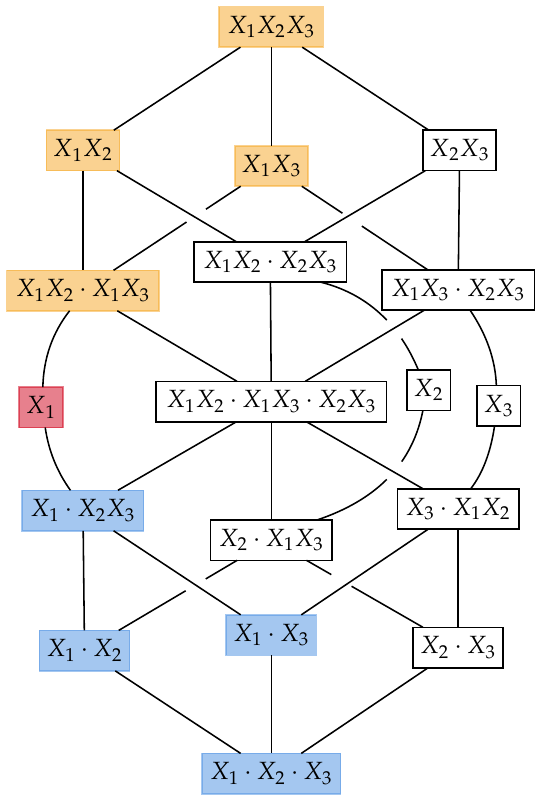}
\caption{PID supports $\mathcal{L}_{\mathrm{MI}}(X_1)$ and $\mathcal{L}_{\mathrm{CMI}}(X_1)$ for $\mathbf{X}=\{X_{1},X_{2},X_{3}\}$. The \mybox{my-yellow}+\mybox{redd} atoms belong to $\mathcal{L}_{\mathrm{CMI}}(X_1)$, whereas the \mybox{bluu}+\mybox{redd} atoms belong to $\mathcal{L}_{\mathrm{MI}}(X_1)$.}
\label{fig_MI&CMI}
\end{center}
\end{figure}

\begin{proof}[Proof of Proposition~\ref{pid-cond2}]
By Proposition~\ref{pid-cond}, a singleton atom $\{B\}$ contributes to $I(X_j;Y\mid\mathbf{Z})$ exactly when $B\subseteq\{X_j\}\cup\mathbf{Z}$ but $B\nsubseteq\mathbf{Z}$. Equivalently, $X_j\in B$ and $B\setminus\{X_j\}\subseteq\mathbf{Z}$. The support is therefore nested as $\mathbf{Z}$ grows, and $|B|=r$ requires at least $r-1$ conditioned sources. This is a set-inclusion statement only: partial atoms may be signed, and redundant contributions can change, so no ordering of the two CMI values follows without additional assumptions.
\end{proof}

\begin{prop}[PID supports of the two conditions]\label{F&M_atoms}
Assume the PID atoms are nonnegative, $I_{\partial}(\alpha;Y)\geq0$ for all $\alpha\in\mathcal{A}(\mathbf{X})$, and consider $X_i\in\mathrm{PA}_Y=\mathbf{X}$. Let $u_i:=I_{\partial}(\{\{X_i\}\};Y)$. Then:
\begin{enumerate}[(a)]
\item \textbf{condition (CF)} obeys
\[
u_i>0
\;\Longrightarrow\;
\text{(CF)}
\;\Longrightarrow\;
\exists\alpha\in\mathcal{L}_{\mathrm{MI}}(X_i):
I_{\partial}(\alpha;Y)>0 ,
\]
and neither implication is reversible from nonnegativity alone;
\item \textbf{condition (CM)} holds if and only if
\[
\exists\alpha\in\mathcal{L}_{\mathrm{CMI}}(X_i):
I_{\partial}(\alpha;Y)>0.
\]
\end{enumerate}
\end{prop}

\begin{proof}
Nonnegativity permits passage between a positive sum and a positive atom in its support. For any $\mathbf{Z}\subseteq\mathbf{X}_{-i}$, Proposition~\ref{pid-cond} places $\{\{X_i\}\}$ in the support of $I(X_i;Y\mid\mathbf{Z})$. Thus $u_i>0$ implies every CMI required by (CF) is positive. Conversely, (CF) at $\mathbf{Z}=\emptyset$ gives $I(X_i;Y)>0$; Eq.~\eqref{mi_cmi_atoms1} then implies that some atom in $\mathcal{L}_{\mathrm{MI}}(X_i)$ is positive.

Neither converse follows from nonnegativity alone. A purely redundant copy can make $I(X_i;Y)>0$ while $I(X_i;Y\mid X_j)=0$. In the other direction, already with two parents, positive redundant and synergistic atoms can keep both required rungs positive even when $u_i=0$; whether a chosen redundancy functional assigns those atoms positive values is measure-dependent.

Finally, (CM) is $I(X_i;Y\mid\mathbf{X}_{-i})>0$. Equation~\eqref{mi_cmi_atoms2} and nonnegativity make this equivalent to a positive atom in $\mathcal{L}_{\mathrm{CMI}}(X_i)$.
\end{proof}

\noindent The asymmetry concerns the conditioning queries: minimality is characterised by one support, whereas (CF) requires positivity for every admissible conditioning set. Positivity of the marginal support alone is insufficient for arbitrary source laws; the independent-source specialisation was stated above.

\section{Simulation details\label{appC}}

We give the experimental and numerical details behind the causal-discovery comparisons and conditioning profiles in the main text. The controls below distinguish structural cancellation from finite-signal and source-dependence effects.

\paragraph{Relation among the numerical families.}
Table~\ref{tab:cda_baselines} samples the full fixed-order family $\mathcal F_k$, with every order-$k$ coefficient drawn independently. Figure~\ref{pure_cmi} restricts the same Gibbs response to one active interaction so that its conditioning threshold is visible without overlap. The left profile in Fig.~\ref{growth_cmi} combines three sparse interactions of different orders under the same response rule, whereas the right profile is the correlated-source sensitivity model of Eq.~\eqref{correlated_extension}. Thus only the last control lies outside the independent-source construction used in Corollary~\ref{cor:pure-order} and Lemma~\ref{lem:minimality-ae}.

\subsection*{Causal-discovery experiments}

\paragraph{Data-generating process.}
For the CDA experiment of the main text, we use Eq.~(\ref{simple_high_ising}) with $n=4$ independent Rademacher sources, target $Y\in\{-1,1\}$, and couplings $J_{\bm{\alpha}}\sim\mathrm{Uniform}(-0.2,0.2)$ independently per trial. We draw $\mathbf{X}$ from its exogenous product distribution and then generate $Y$ according to
\[
P_J(y\mid\mathbf{x})=
\frac{\exp[\beta y g_J(\mathbf{x})]}{2\cosh[\beta g_J(\mathbf{x})]},
\qquad y\in\{-1,1\}.
\]
Two observation formats are generated from the same causal mechanism:
\begin{itemize}
  \item \textbf{i.i.d. rows} $(\mathbf{X},Y)$ for PC and GES;
  \item \textbf{lagged rows} $(\mathbf{X}(t-1),Y(t))$ for PCMCI, with fresh independent sources at each $t$.
\end{itemize}
Results are averaged over 200 independent trials per condition ($k\in\{2,3,4,5\}$, $N\in\{500,1000,2000,5000\}$, $\alpha\in\{0.05,0.01\}$). Couplings are matched across sample sizes within each $(k,\text{trial})$ pair, whereas observational samples are independently seeded. Trial-level percentile intervals use 5000 bootstrap resamples.
PC and PCMCI use the likelihood-ratio $G^2$ test for categorical variables. Here, $G^2$ compares observed counts with those expected under conditional independence, within each configuration of the conditioning variables. GES uses the discrete BDeu local score. Thus all three baselines are matched to the binary data. The significance level $\alpha$ applies to PC and PCMCI; the two GES columns repeat the same score-based run.

\paragraph{Full sample-size sweep.}
We report the full sweep $N\in\{500,1000,2000,5000\}$ in Table~\ref{tab:cda_baselines}, allowing recovery at each interaction order to be followed as the sample grows.

Trial-level bootstrap intervals are computed from the same 200 trials used for the reported means.

\paragraph{Control matched for total target information.}
We now ask whether the poor recovery could simply reflect a weaker total signal. To separate this possibility from structural marginal cancellation, we paired each $k=3$ coupling draw above with a $k=2$ direction drawn independently and rescaled it until exact enumeration gave the same $I(\mathbf X;Y)$ in both models (absolute mismatch below $10^{-10}$ bits per trial). We then generated $N=5000$ observations from each model with the same random stream and ran the same PC implementation with the likelihood-ratio test described above. Across 200 paired trials, mean total target information was $0.052$ bits in both arms. The mean minimum full-parent CMI was $0.015$ bits for $k=3$ and $0.002$ bits for the matched $k=2$ arm. The odd-order model therefore had a stronger full-parent signal on average, yet its marginal information vanished identically and PC recovered far fewer parents. The control fixes total target information, not each parent-wise margin.

\paragraph{The even case $k=4$.}\label{app:k4_scaling}
For $k=4$, $g_J(-\mathbf{x})=-g_J(\mathbf{x})$, so Lemma~\ref{lem:lema1} does not apply and no exact marginal cancellation is claimed. Exact enumeration of the causal DGP at the coupling scale used in Tab.~\ref{tab:cda_baselines} gives representative marginal-information values of order $10^{-5}$ (maximum $1.72\times10^{-5}$ in the diagnostic draw), while the minimum full-parent CMI is $1.82\times10^{-2}$. Thus $k=4$ is a practical low-signal comparison over $N\leq5000$, not an extrapolation of the odd-$k$ theorem.

\paragraph{Scope of the GES result.}\label{app:ges_score_dependence}
The experiment uses standard GES with the discrete BDeu local score. Its forward phase considers single-edge additions. In the odd pure-order DGP, all sources are mutually independent and each source is marginally independent of $Y$; consequently no first single-edge addition into $Y$ improves the population local score, and the forward search stalls. This does not establish an algorithm-agnostic impossibility: a coordinated multi-parent move, exhaustive search, or a hypergraph model can access the joint contingency. Our claim is therefore restricted to the tested implementation and its search neighbourhood.

\paragraph{PCMCI estimator details.}\label{app:pcmci_details}
All PCMCI evaluations use the categorical likelihood-ratio $G^2$ test. Its first parent-selection rung tests each lagged source without other candidate parents, so every odd-order parent is null at that population screen by Lemma~\ref{lem:lema1}. At fixed finite-sample $\alpha$, however, a null parent can survive by chance. For $k=3$, conditioning on one such seed exposes the pair-product terms containing it, allowing the final MCI stage to recover other parents. At $N=5000$, the observed probabilities of at least one PC1 seed were $0.210$ and $0.035$ for $\alpha=0.05$ and $0.01$, respectively, close to the independent-test references $1-(1-\alpha)^4=0.185$ and $0.039$. Conditional mean recall was $0.0095$ without a seed versus $0.863$ with one at $\alpha=0.05$, and $0$ versus $0.857$ at $\alpha=0.01$. A single seed does not expose the higher conditioning order required at $k=5$, whose recall remains near the nominal floor. This finite-sample escape clarifies that the result concerns the tested search schedule, not every procedure capable of joint parent entry.

\subsection*{Conditioning profiles}
\paragraph{Exact enumeration.}
All profiles are computed by exact enumeration at $\beta=1$ over 40 independent coupling draws. At each rung we evaluate Eq.~\eqref{eq:cmi-profile}: $\mathbf S_r$ contains $r$ sources and, for each $X_i\in\mathbf S_r$, the conditioning set is the leave-one-out set $\mathbf S_r\setminus\{X_i\}$ of fixed size $r-1$. We then average $M_r$ uniformly over all $r$-source subsystems, avoiding a preferred ordering of the source labels. Figure~\ref{pure_cmi} uses $n=4$ independent Rademacher sources and one isolated order-$k$ interaction, with its nonzero coefficient drawn from $\mathrm{Uniform}(-0.2,0.2)$. The normalized means for $k=2,3,4,5$ are, respectively, $(1/4,1/2,3/4,1)$, $(0,1/3,1,2)$, $(0,0,3/4,3)$, and $(0,0,0,4)$. Each profile is normalized within its coupling draw before averaging.

Figure~\ref{growth_cmi} uses $n=9$. The synergy-dominant model has three disjoint source blocks of sizes two, three, and four, with coefficients drawn from $\mathrm{Uniform}(-0.2,0.2)$; uniform subsystem averaging yields the smooth monotone growth from zero to $3.09$. The redundancy-dominant profile uses Eq.~\eqref{correlated_extension} with the same ferromagnetic pair coupling for every pair, $K_{ij}=K\sim\mathrm{Uniform}(0.18,0.22)$, and the linear target field $g_\lambda(\mathbf x)=\lambda\sum_i x_i$, $\lambda\sim\mathrm{Uniform}(0.15,0.25)$. Over this finite-coupling window, the normalized mean decreases monotonically from $0.72$ to $3.9\times10^{-2}$; the unnormalized full rung is $(2.7\pm0.7)\times10^{-2}$ bits. The smallest parent-wise full CMI is $1.7\times10^{-3}$ bits, so all 40 draws satisfy the causal-minimality criterion for the arrows into $Y$, while its margin is small relative to the marginal rung.

\paragraph{Analytic redundancy control.}
The identities in Proposition~\ref{prop:redundancy-margin} were also checked by direct enumeration of $P(X_1,X_2,Y)$ over a grid of $(K,\lambda)$; the closed forms and enumerated values agree to machine precision.

\paragraph{Minimality margin across regimes.}\label{app:premise_plane}
We evaluate $m_{\min}=\min_i I(X_i;Y\mid\mathbf{X}_{-i})$ over active parents for the same 40 draws as Figs.~\ref{pure_cmi} and \ref{growth_cmi}. The causal-minimality criterion for every arrow into $Y$ holds in every draw. The correlated-copy endpoint remains positive at finite coupling while becoming negligible relative to its bottom rung.

\subsection*{Illustration for the conclusion}
For the isolated-product example discussed in the conclusion, we consider the one-interaction subfamily $g_J(\mathbf x)=J_B\prod_{i\in B}x_i$ and let $Z_B=\prod_{i\in B}X_i$. Because $Y$ is binary and $\mathbb E[Y\mid\mathbf X]=\tanh(\beta J_B Z_B)$, the target law depends on $\mathbf X$ only through $Z_B$: $I(\mathbf X;Y\mid Z_B)=0$. Independent Rademacher sources make $Z_B$ unbiased, so $I(Z_B;Y)=1-h_2[(1+\tanh(\beta J_B))/2]>0$ whenever $J_B\ne0$. We supply the active block $B$ in this example; neither its identification from data nor a general map-recovery theorem follows.

\subsection*{Reproducibility and numerical tables}
\paragraph{Reproducibility.}
The simulation code and numerical data underlying the figures and tables will be deposited in a public repository. Before release, they are available from the corresponding author upon reasonable request.

\begingroup
\setlength{\intextsep}{6pt}
\begin{table}[H]
\centering
\small
\renewcommand{\arraystretch}{1.12}
\setlength{\tabcolsep}{7pt}
\begin{tabular}{cc@{\hspace{10pt}}ccc@{\hspace{16pt}}ccc}
\hline
& & \multicolumn{3}{c}{$\alpha=0.05$} & \multicolumn{3}{c}{$\alpha=0.01$} \\
\cline{3-5}\cline{6-8}
$k$ & $N$ & PC & GES & PCMCI & PC & GES & PCMCI \\
\hline
\multirow{4}{*}{2}
 & 500  & 0.430 & 0.291 & 0.485 & 0.328 & 0.291 & 0.380 \\
 & 1000 & 0.584 & 0.489 & 0.609 & 0.510 & 0.489 & 0.504 \\
 & 2000 & 0.695 & 0.630 & 0.706 & 0.636 & 0.630 & 0.649 \\
 & 5000 & 0.811 & 0.761 & 0.819 & 0.769 & 0.761 & 0.776 \\
\hline
\multirow{4}{*}{3}
 & 500  & 0.064 & 0.004 & 0.119 & 0.011 & 0.004 & 0.023 \\
 & 1000 & 0.045 & 0.006 & 0.165 & 0.010 & 0.006 & 0.019 \\
 & 2000 & 0.040 & 0.005 & 0.179 & 0.006 & 0.005 & 0.016 \\
 & 5000 & 0.058 & 0.000 & 0.189 & 0.010 & 0.000 & 0.030 \\
\hline
\multirow{4}{*}{4$^\dagger$}
 & 500  & 0.044 & 0.001 & 0.048 & 0.006 & 0.001 & 0.015 \\
 & 1000 & 0.033 & 0.000 & 0.051 & 0.004 & 0.000 & 0.003 \\
 & 2000 & 0.046 & 0.001 & 0.058 & 0.014 & 0.001 & 0.010 \\
 & 5000 & 0.044 & 0.001 & 0.056 & 0.005 & 0.001 & 0.009 \\
\hline
\multirow{4}{*}{5}
 & 500  & 0.048 & 0.001 & 0.055 & 0.014 & 0.001 & 0.014 \\
 & 1000 & 0.033 & 0.001 & 0.063 & 0.004 & 0.001 & 0.010 \\
 & 2000 & 0.036 & 0.001 & 0.060 & 0.004 & 0.001 & 0.006 \\
 & 5000 & 0.045 & 0.000 & 0.060 & 0.011 & 0.000 & 0.016 \\
\hline
\multicolumn{8}{l}{\footnotesize $^\dagger$Lemma~\ref{lem:lema1} does not apply; exact enumeration gives small but nonzero marginal MI.}
\end{tabular}
\caption{Parent recall (mean over 200 trials) for the pure-order causal model. Odd $k$ are covered by Lemma~\ref{lem:lema1}; $k=4$ is a finite-signal comparison outside its scope.}
\label{tab:cda_baselines}
\end{table}

\begin{table}[H]
\centering
\small
\begin{tabular}{cccc}
\hline
$\alpha$ & $k=3$ & matched $k=2$ & paired difference (95\% CI) \\
\hline
$0.05$ & $0.058$ & $0.861$ & $0.804\;(0.774,0.833)$ \\
$0.01$ & $0.010$ & $0.829$ & $0.819\;(0.794,0.843)$ \\
\hline
\end{tabular}
\caption{Parent recall for PC in the $k=3$ family and a per-trial, total-information-matched $k=2$ control ($n=4$, $N=5000$, 200 paired trials). The difference is $k=2$ minus $k=3$; percentile 95\% intervals use 5000 paired bootstrap resamples.}
\label{tab:matched_signal_pc}
\end{table}

\begin{table}[H]
\centering
\begin{tabular}{lcc}
\hline
Family & mean $m_{\min}$ & smallest draw \\
\hline
isolated $k=2$               & $0.008\pm0.009$ & $2\times10^{-6}$ \\
isolated $k=3$               & $0.009\pm0.009$ & $3\times10^{-5}$ \\
isolated $k=4$               & $0.009\pm0.009$ & $6\times10^{-5}$ \\
isolated $k=5$               & $0.009\pm0.008$ & $1\times10^{-5}$ \\
sparse exclusive $k=3,4,5$  & $0.001\pm0.002$ & $2\times10^{-6}$ \\
pairwise correlated copy     & $0.003\pm0.001$ & $2\times10^{-3}$ \\
\hline
\end{tabular}
\caption{Causal-minimality margin for arrows into $Y$, in bits. Uncertainties are standard deviations over 40 coupling draws.}
\label{tab:minimality_margin}
\end{table}

\endgroup

\fi
\end{document}